\documentclass[12pt,draftclsnofoot,onecolumn]{IEEEtran}

\makeatletter
\newcommand{\vast}{\bBigg@{4}}
\makeatother
\usepackage{multirow}
\usepackage{graphicx}
\usepackage{epstopdf}
\usepackage[cmex10]{amsmath}
\usepackage{amssymb}

\usepackage{algpseudocode}
\usepackage{algorithm}
\usepackage{amsthm}
\theoremstyle{definition}
\theoremstyle{remark}

\newtheorem{lmm}{Lemma}

\newtheorem*{prop}{Proposition}

\makeatletter
\g@addto@macro\th@remark{\thm@headpunct{\normalfont:}}
\makeatother
\makeatletter
\newcommand{\distas}[1]{\mathbin{\overset{#1}{\kern\z@\sim}}}%
\newsavebox{\mybox}\newsavebox{\mysim}
\newcommand{\distras}[1]{%
  \savebox{\mybox}{\hbox{\kern3pt$\scriptstyle#1$\kern3pt}}%
  \savebox{\mysim}{\hbox{$\sim$}}%
  \mathbin{\overset{#1}{\kern\z@\resizebox{\wd\mybox}{\ht\mysim}{$\sim$}}}%
}

\begin{document}

\title{On the Ergodic Capacity of Underlay Cognitive Dual-Hop AF Relayed Systems under Non-Identical Generalized-\emph{K} Fading Channels}

\author{Nikolaos~I.~Miridakis
\thanks{N. I. Miridakis is with the Department of Computer Systems Engineering, Piraeus University of Applied Sciences, 122 44, Aegaleo, Greece (e-mail: nikozm@unipi.gr).}
}

\markboth{}%
{On the Ergodic Capacity of Underlay Cognitive Dual-Hop AF Relayed Systems under Non-Identical Generalized-\emph{K} Fading Channels}

\maketitle

\begin{abstract}
The ergodic capacity of underlay cognitive (secondary) dual-hop relaying systems is analytically investigated. Specifically, the amplify-and-forward transmission protocol is considered, while the received signals undergo multipath fading and shadowing with non-identical statistics. To efficiently describe this composite type of fading, the well-known generalized-$K$ fading model is used. New analytical expressions and quite accurate closed-form approximations regarding the ergodic capacity of the end-to-end communication are obtained, in terms of finite sum series of the Meijer's-$G$ function. The analytical results are verified with the aid of computer simulations, while useful insights are revealed. 
\end{abstract}


\IEEEpeerreviewmaketitle

\section{Introduction}
\IEEEPARstart{U}{nderlay} cognitive transmission represents one of the most popular spectrum sharing techniques, where secondary (unlicensed) users utilize the spectrum resources of another primary (licensed) service. Due to its mode of operation, the transmission power of secondary users is limited, such that its interference onto the primary users remains below prescribed tolerable levels. However, this dictated constraint dramatically affects the coverage and/or capacity of the secondary communication. Such a condition can be effectively counteracted with the aid of relayed transmission. Performance assessment of these systems has been well-investigated in the open technical literature to date (e.g., see \cite{ref2} and references therein). Nevertheless, these works assumed non-shadowing environments; a rather infeasible condition.

In practical wireless communication systems, the signal always experiences composite small-scale (multipath) fading and large-scale shadowing simultaneously. The rigorous log-normal distribution appropriately describes the latter effect, giving rise to composite models, such as the Rayleigh/log-normal distribution. Alternatively, the generalized-$K$ ($K_{G}$) distribution model can efficiently describe this composite effect, while preserves mathematical tractability at the same time. It is noteworthy that it includes the classical Rayleigh, Nakagami-$m$ and Rayleigh/Gamma (i.e., the $K$ distribution) fading models as special types \cite{ref3,meegc}.

Yet, only few research works have investigated the performance of cognitive relayed transmission over composite fading/shadowing channels. Specifically, the performance of an underlay cognitive relaying system was analytically studied in \cite{ref9} under $K_{G}$ fading channels, by considering the decode-and-forward (DF) relaying scheme. The derived expressions therein were provided in terms of an infinite series representation. Further, the authors of \cite{ref10} studied the scenario of underlay cognitive systems with multi-hop/multi-relay transmission under $K_{G}$ fading channels, when the amplify-and-forward (AF) relaying scheme is used. However, the end-to-end ($e2e$) performance was only approximated in that work with the aid of bound expressions, not exact ones. Notably, these bounds were tight only in high signal-to-noise ratio (SNR) regions (i.e., see \cite[Eq. (6)]{ref10}). Nonetheless, it should be noticed that the mode of operation used for underlay cognitive systems fundamentally supports quite a low transmit power, which contraindicates the accuracy of the former bounds in low SNR regions. In addition, identically distributed fading channels were assumed in \cite{ref10} with common statistics (i.e., equal distance and fading severity for all the included signals), which is not always the case in real-life network setups.  

Capitalizing on the aforementioned observations, an underlay cognitive relayed system is investigated in current work, where the relay utilizes the cost-effective and computational-efficient AF transmission scheme. The signal of each link (from both secondary and primary nodes) is subject to independent and non-identical composite multipath fading/shadowing, modeled by the $K_{G}$ distribution. This condition is suitable for most practical applications where the involved signals undergo arbitrary link distances with distinct fading statistics. New analytical expressions and quite accurate closed-form approximations for the $e2e$ ergodic capacity are derived, while some useful engineering insights are also obtained. The derived expressions are valid in the entire SNR region (low-to-high), while they are time-efficient in comparison to other existing methods so far (e.g., numerical manifold integrations or Monte-Carlo simulations).

\section{System Model}
Consider an underlay (secondary) dual-hop system where the source (S) communicates with the destination (D) via an intermediate relay (R) node. This system operates in the presence of a licensed (primary) node (P$_{R}$). Moreover, the signal transmission power of the secondary system is, in principle, maintained quite low in order not to dramatically affect the reception quality of the primary communication in terms of interfering power. To this end, assume that the direct communication between S and D is not feasible due to strong propagation attenuation and/or severe shadowing, whereas keeping in mind the constrained transmission power regime. Hence, the $e2e$ communication is facilitated with the aid of R.

\subsection{Power Allocation}
The transmitted power of S and R are denoted as $P_{S}=w/\left|h_{SP_{R}}\right|^{2}$ and $P_{R}=w/\left|h_{RP_{R}}\right|^{2}$, where $\left|h_{SP_{R}}\right|^{2}$ and $\left|h_{RP_{R}}\right|^{2}$ correspond to the channel gains of S-to-P$_{R}$ and R-to-P$_{R}$, respectively. Also, $w$ represents a power threshold, the so-called \emph{interference temperature}, which should not be exceeded during the transmission of secondary nodes. Such an approach has been widely adopted in the open technical literature (e.g., see \cite{ref2}-\cite{ref10} and references therein), mainly because it effectively balances performance and complexity.

In principle, channel state information (CSI) of the links between the primary and secondary nodes can be obtained through a feedback channel from the primary service and due to the channel reciprocity. CSI can also be captured through a band manager that mediates the exchange of information between the primary and secondary networks \cite{ref2}. It is noteworthy that when $w/\left|h_{SP_{R}}\right|^{2}$ (or $w/\left|h_{RP_{R}}\right|^{2}$) happens to be higher than the maximal allowable transmitted power, say $P_{\text{max}}$, power control of the corresponding secondary node may modify $w$ to $w'$ so as $w'/\left|h_{SP_{R}}\right|^{2}=P_{\text{max}}$ (or $w'/\left|h_{RP_{R}}\right|^{2}=P_{\text{max}}$) is satisfied. This issue is further analyzed in the next section.

\subsection{Signal Model}
The $e2e$ communication occurs in two consecutive transmission phases, one for each hop. The received signal of R at the end of the first phase is given by $y_{R}=h_{SR}x+n_{R}$, where $y_{R}$, $h_{SR}$, $x$ and $n_{R}$ represent the received signal, the channel coefficient of the S-to-R link, the transmitted signal and the additive white Gaussian noise (AWGN) at R, respectively. Then, this signal is amplified with the variable gain $G$ and forwarded to D during the second transmission phase. Hence, the overall signal at D is expressed as $y_{D}=Gh_{RD}y_{R}+n_{D}=Gh_{SR}h_{RD}x+Gh_{RD}n_{R}+n_{D}$, where $h_{RD}$ and $n_{D}$ correspond to the channel coefficient of the R-to-D link and the AWGN at D, respectively. For notational simplicity and without loss of generality, assume that noise powers are identical, i.e., $P_{n_{R}}=P_{n_{D}}\triangleq N_{0}$.

We retain our focus on CSI-assisted AF transmission, since the knowledge of CSI represents a requisite for the efficient operation of underlay cognitive relaying. Thus, since $P_{R}=G^{2}(\left|h_{SR}\right|^{2}P_{S}+N_{0})$, it can be seen that $G^{2}=(|h_{RP_{R}}|^{2}(\frac{|h_{SR}|^{2}}{|h_{SP_{R}}|^{2}}+\frac{N_{0}}{w}))^{-1}$. After some straightforward algebra, the $e2e$ SNR reads as
\begin{align}
\gamma_{e2e}=\frac{G^{2}\left|h_{RD}\right|^{2}\left|h_{SR}\right|^{2}P_{S}}{G^{2}\left|h_{RD}\right|^{2}N_{0}+N_{0}}=\frac{\gamma_{1}\gamma_{2}}{\gamma_{1}+\gamma_{2}+1},
\end{align}
where $\gamma_{1}\triangleq \frac{w\left|h_{SR}\right|^{2}}{N_{0}\left|h_{SP_{R}}\right|^{2}}$ and $\gamma_{2}\triangleq \frac{w\left|h_{RD}\right|^{2}}{N_{0}\left|h_{RP_{R}}\right|^{2}}$.

\section{Performance Analysis}
Ergodic capacity of the entire $e2e$ communication is defined as
\begin{align}
\nonumber
C\triangleq \frac{1}{2}\mathbb{E}[\log_{2}(1+\gamma_{e2e})]&=\frac{1}{2}\mathbb{E}\left[\log_{2}\left(\frac{(1+\gamma_{1})(1+\gamma_{2})}{1+\gamma_{1}+\gamma_{2}}\right)\right]\\
&=\frac{1}{2}(C_{1}+C_{2}-C_{1+2}),
\label{ergcap}
\end{align}
where $\mathbb{E}[.]$ denotes expectation, $C_{l}\triangleq \mathbb{E}[\log_{2}(1+\gamma_{l})]$ with $l\in \{1,2\}$, $C_{1+2}\triangleq \mathbb{E}[\log_{2}(1+\gamma_{1}+\gamma_{2})]$ and the factor $1/2$ is due to the involvement of two transmission phases.

Following lemmas will be quite useful for the statistical derivation of (\ref{ergcap}).

\begin{lmm}
Probability density function (PDF) of $\gamma_{l}$ with $l \in \{1,2\}$ is obtained by
\begin{align}
\nonumber
&f_{\gamma_{l}}(x)=\frac{\left(\frac{k_{il}m_{il}\Omega_{jl}N_{0}}{k_{jl}m_{jl}\Omega_{il}w}\right)^{\Delta}x^{\Delta-1}}{\Gamma(m_{il})\Gamma(m_{jl})\Gamma(k_{il})\Gamma(k_{jl})}\\
&\times G^{2,2}_{2,2}\left[\frac{k_{il}m_{il}\Omega_{jl}N_{0}x}{k_{jl}m_{jl}\Omega_{il}w}~\vline
\begin{array}{c}
1-\Delta-k_{jl},1-\Delta-m_{jl}\\
\Theta,-\Theta
\end{array} \right],
\label{fg}
\end{align}
where $\{i,j\}=\{SR,SP_{R}\}$ for $l=1$, while $\{i,j\}=\{RD,RP_{R}\}$ for $l=2$. Also, $G^{m,n}_{p,q}[.]$ stands for the Meijer's-$G$ function \cite[Eq. (9.301)]{tables}. Moreover, $m_{il}\geq 0.5$ and $k_{il}\geq 0$ denote the multipath fading and shadowing severity of the $i$th channel coefficient, respectively, whereas higher (lower) values indicate less (more) severe fading channel conditions. In addition, $\Delta\triangleq \frac{k_{il}+m_{il}}{2}$ and $\Theta\triangleq \frac{k_{il}-m_{il}}{2}$ are introduced for notational simplicity, while $\Omega_{il}\triangleq d^{-\alpha_{il}}_{il}$, where $d_{il}$ and $\alpha_{il}$ represent the corresponding normalized link distance (with a reference distance equal to $1$ km) and path loss factor. Usually $\alpha_{il}\in \{2-6\}$, with $\alpha_{il}=2$ indicating free-space loss, while $\alpha_{il}>2$ representing suburban to dense urban environments. 
\end{lmm}

\begin{proof}
Since $h_{il}$ is $K_{G}$ distributed, $\left|h_{il}\right|^{2}$ has a PDF as follows
\begin{align}
f_{\left|h_{il}\right|^{2}}(x)=\frac{\left(\frac{k_{il}m_{il}}{\Omega_{il}}\right)^{\Delta}x^{\Delta-1}}{\Gamma(m_{il})\Gamma(k_{il})}G^{2,0}_{0,2}\left[\frac{k_{il}m_{il}x}{\Omega_{il}}~\vline
\begin{array}{c}
-\\
\Theta,-\Theta
\end{array} \right].
\label{fh}
\end{align}
For the derivation of (\ref{fh}), the transformation \cite[Eq. (9.34.3)]{tables} is used into \cite[Eq. (2)]{ref3}. Further, it holds that $f_{\left|h_{il}\right|^{2}/\left|h_{jl}\right|^{2}}(x)=\int^{\infty}_{0}y f_{\left|h_{il}\right|^{2}}(xy)f_{\left|h_{jl}\right|^{2}}(y)dy$, since $\left|h_{il}\right|^{2}$ and $\left|h_{jl}\right|^{2}$ are mutually independent, while $f_{\left|h_{jl}\right|^{2}}(.)$ is provided in (\ref{fh}) by substituting $i$ with $j$. Finally, utilizing \cite[Eq. (2.24.1.1)]{ref4} into the latter integral and after some straightforward manipulations concludes the proof.
\end{proof}

\begin{lmm}
Moment generating function (MGF) of $\gamma_{l}$ yields as
\begin{align}
\nonumber
&\mathcal{M}_{\gamma_{l}}(s)=\frac{\left(\frac{k_{il}m_{il}\Omega_{jl}N_{0}}{k_{jl}m_{jl}\Omega_{il}w s}\right)^{\Delta}}{\Gamma(m_{il})\Gamma(m_{jl})\Gamma(k_{il})\Gamma(k_{jl})}\\
&\times G^{2,3}_{3,2}\left[\frac{k_{il}m_{il}\Omega_{jl}N_{0}}{k_{jl}m_{jl}\Omega_{il}w s}~\vline
\begin{array}{c}
1-\Delta,1-\Delta-k_{jl},1-\Delta-m_{jl}\\
\Theta,-\Theta
\end{array} \right]
\label{mg}
\end{align}
\end{lmm}

\begin{proof}
Since $\mathcal{M}_{\gamma_{l}}(s)\triangleq \mathbb{E}[\exp(-s \gamma_{l})]$, the following integral appears
\begin{align}
\nonumber
&\int^{\infty}_{0}x^{\Delta-1}\exp(-s x)\\
&\times G^{2,2}_{2,2}\left[\frac{k_{il}m_{il}\Omega_{jl}N_{0}x}{k_{jl}m_{jl}\Omega_{il}w}~\vline
\begin{array}{c}
1-\Delta-k_{jl},1-\Delta-m_{jl}\\
\Theta,-\Theta
\end{array} \right]dx,
\label{intmg}
\end{align}
which is evaluated in closed-form with the aid of \cite[Eq. (2.24.3.1)]{ref4} and (\ref{mg}) directly arises.
\end{proof}

\begin{lmm}
Partial derivative of $\mathcal{M}_{\gamma_{l}}(.)$ is derived by
\begin{align}
\nonumber
&\frac{\partial}{\partial s}\mathcal{M}_{\gamma_{l}}(s)=-\frac{\left(\frac{k_{il}m_{il}\Omega_{jl}N_{0}}{k_{jl}m_{jl}\Omega_{il}w}\right)^{\Delta}}{\Gamma(m_{il})\Gamma(m_{jl})\Gamma(k_{il})\Gamma(k_{jl})s^{\Delta+1}}\\
&\times G^{2,3}_{3,2}\left[\frac{k_{il}m_{il}\Omega_{jl}N_{0}}{k_{jl}m_{jl}\Omega_{il}w s}~\vline
\begin{array}{c}
-\Delta,1-\Delta-k_{jl},1-\Delta-m_{jl}\\
\Theta,-\Theta
\end{array} \right].
\label{dmg}
\end{align}
\end{lmm}

\begin{proof}
From (\ref{intmg}), the following integral appears
\begin{align}
\nonumber
&\frac{\partial}{\partial s}\mathcal{M}_{\gamma_{l}}(s)\propto -\int^{\infty}_{0}x^{\Delta}\exp(-s x)\\
&\times G^{2,2}_{2,2}\left[\frac{k_{il}m_{il}\Omega_{jl}N_{0}x}{k_{jl}m_{jl}\Omega_{il}w}~\vline
\begin{array}{c}
1-\Delta-k_{jl},1-\Delta-m_{jl}\\
\Theta,-\Theta
\end{array} \right]dx.
\label{intmgg}
\end{align}
Again, using \cite[Eq. (2.24.3.1)]{ref4}, as for the derivation of (\ref{mg}), (\ref{dmg}) is extracted.
\end{proof}

We are now in a position to formulate the $e2e$ ergodic capacity.

\begin{prop}
Ergodic capacity of the considered system is expressed as
\begin{align}
\nonumber
&C=\sum^{2}_{l=1}\frac{\left(\frac{k_{il}m_{il}\Omega_{jl}N_{0}}{k_{jl}m_{jl}\Omega_{il}w}\right)^{\Delta}}{\log(4)\Gamma(m_{il})\Gamma(m_{jl})\Gamma(k_{il})\Gamma(k_{jl})}\\
\nonumber
&\times G^{4,3}_{4,4}\left[\scriptstyle\frac{k_{il}m_{il}\Omega_{jl}N_{0}}{k_{jl}m_{jl}\Omega_{il}w}~\vline
\begin{array}{c}
\scriptstyle 1-\Delta-k_{jl},1-\Delta-m_{jl},-\Delta,1-\Delta\\
\scriptstyle \Theta,-\Theta,-\Delta,-\Delta
\end{array} \right]-\sum^{\mathcal{N}}_{v=1}\psi_{v}\\
&\times \frac{\text{Ei}(-s_{v})}{\log(4)}\left(\textstyle \mathcal{M}_{\gamma_{1}}(s_{v})\frac{\partial}{\partial s_{v}}\mathcal{M}_{\gamma_{2}}(s_{v})+\mathcal{M}_{\gamma_{2}}(s_{v})\frac{\partial}{\partial s_{v}}\mathcal{M}_{\gamma_{1}}(s_{v})\right)
\label{allc}
\end{align}
where $\text{Ei(.)}$ stands for the exponential integral function Ei \cite[Eq. (8.211.1)]{tables}, $\psi_{v}\triangleq \frac{\pi^{2}\sin(\frac{(2v-1)\pi}{2\mathcal{N}})}{4\mathcal{N}\cos^{2}(\frac{\pi}{4}\cos(\frac{(2v-1)\pi}{2\mathcal{N}})+\frac{\pi}{4})}$ and $s_{v}\triangleq \tan(\frac{\pi}{4}\cos(\frac{(2v-1)\pi}{2\mathcal{N}})+\frac{\pi}{4})$. Ideally, $\mathcal{N}$ goes to infinity, but as was shown into \cite{ref6,ref7} (and it will be verified from the subsequent numerical results), setting $\mathcal{N}=60$ provides quite a high accuracy level.
\end{prop}

\begin{proof}
Based on (\ref{ergcap}), $C_{l}$ is obtained by using the transformation of the logarithm function into the Meijer's-$G$ function \cite[Eq. (8.4.6.5)]{ref4} and then utilizing \cite[Eq. (2.24.1.1)]{ref4}. Unfortunately, this standard technique can not be used for the derivation of $C_{1+2}$, since the corresponding PDF of $\gamma_{1}+\gamma_{2}$ is not feasible. Nonetheless, based on \cite[Eq. (10)]{ref6}, we are able to bypass the PDF involvement and to evaluate $C_{1+2}$ as
\begin{align}
\nonumber
C_{1+2}\triangleq &\frac{1}{\log(4)}\int^{\infty}_{0}\text{Ei}(-s)\\
&\times \left(\mathcal{M}_{\gamma_{1}}(s)\frac{\partial}{\partial s}\mathcal{M}_{\gamma_{2}}(s)+\mathcal{M}_{\gamma_{2}}(s)\frac{\partial}{\partial s}\mathcal{M}_{\gamma_{1}}(s)\right)ds.
\label{capintt}
\end{align}
Still, (\ref{capintt}) cannot be resolved in closed-form in terms of standard build-in functions in well-known mathematical software tools. Thus, using an alternative representation of (\ref{capintt}) with respect to the Gauss-Chebyshev quadrature formula \cite[Eq. (8)]{ref6}, (\ref{capintt}) becomes
\begin{align}
\nonumber
&C_{1+2}\triangleq \frac{1}{\log(4)}\sum^{\mathcal{N}}_{1}\psi_{v}\text{Ei}(-s_{v})\\
&\times \left(\mathcal{M}_{\gamma_{1}}(s_{v})\frac{\partial}{\partial s_{v}}\mathcal{M}_{\gamma_{2}}(s_{v})+\mathcal{M}_{\gamma_{2}}(s_{v})\frac{\partial}{\partial s_{v}}\mathcal{M}_{\gamma_{1}}(s_{v})\right),
\label{capinttt}
\end{align}
yielding (\ref{allc}).
\end{proof}

Hence, plugging (\ref{mg}) and (\ref{dmg}) into (\ref{allc}), a closed-form approximation of the $e2e$ ergodic capacity is obtained in terms of finite sum series of the Meijer's-$G$ function. Notice that $G^{m,n}_{p,q}[.]$ is included as standard build-in function in most popular mathematical software packages; thereby it can be calculated quite easily and efficiently for arbitrary fading/shadowing parameters.

\subsection*{Special case: When transmission power equals $P_{\text{max}}$}
Let $d^{\alpha_{jk}}_{jk}\rightarrow \infty$, where the $k$th hop indicates the asymptotic secondary-to-primary node link (i.e., $k \in \{1,2\}$, $k\neq l$, conditioned on $\Omega_{jk}\rightarrow 0^{+}$). The physical meaning of such a scenario relies on the fact that when the distance of S-to-P$_{R}$ or R-to-P$_{R}$ link (or both) is considerably long, this assumption takes place. Alternatively, this asymptotic behavior can also be modeled for relatively shorter distances, but in dense propagation environments (e.g., urban terrestrials where $\alpha_{jl} > 4$). Hence, $\frac{w}{|h_{jk}|^{2}}\rightarrow \infty$, reflecting that $\frac{w}{|h_{jk}|^{2}}\gg \frac{w'}{|h_{jk}|^{2}}=P_{\text{max}}$. In other words, the conventional (non-cognitive) transmission approach occurs in this case.

To this end, referring back to (\ref{ergcap}) and based on \cite[Eq. (7)]{ref3}, $C_{k}$ becomes
\begin{align}
C_{k}=\frac{\left(\frac{k_{ik}m_{ik}N_{0}}{\Omega_{ik}P_{\text{max}}}\right)^{\Delta}}{\log(4)\Gamma(m_{ik})\Gamma(k_{ik})}G^{4,1}_{2,4}\left[\scriptstyle\frac{k_{ik}m_{ik}N_{0}}{\Omega_{ik}P_{\text{max}}}~\vline
\begin{array}{c}
\scriptstyle -\Delta,1-\Delta\\
\scriptstyle \Theta,-\Theta,-\Delta,-\Delta
\end{array} \right].
\label{ck}
\end{align}
For the more challenging $C_{1+2}$ parameter in (\ref{ergcap}) (which becomes in this case $C_{l+k}$ or $C_{k+l}$), we have from \cite[Eq. (4)]{ref3} that 
\begin{align}
\nonumber
\mathcal{M}_{\gamma_{k}}(s)=&\left(\frac{k_{ik}m_{ik}N_{0}}{\Omega_{ik}P_{\text{max}} s}\right)^{\Delta-\frac{1}{2}}\exp\left(\frac{k_{ik}m_{ik}N_{0}}{2\Omega_{ik}P_{\text{max}} s}\right)\\
&\times \mathcal{W}_{-\Delta+\frac{1}{2},\Theta}\left(\frac{k_{ik}m_{ik}N_{0}}{\Omega_{ik}P_{\text{max}} s}\right),
\label{mk}
\end{align}
where $\mathcal{W}_{\alpha,\beta}(.)$ stands for the Whittaker's-$W$ hypergeometric function \cite[Eq. (9.220.4)]{tables}. Finally, following similar lines of reasoning as for the derivation of (\ref{dmg}), partial derivative of $\mathcal{M}_{\gamma_{k}}(.)$ yields as
\begin{align}
\nonumber
\frac{\partial}{\partial s}\mathcal{M}_{\gamma_{k}}(s)=&-\frac{\left(\frac{k_{ik}m_{ik}N_{0}}{\Omega_{ik}P_{\text{max}}}\right)^{\Delta-\frac{1}{2}}k_{il}m_{il}}{s^{\Delta+\frac{1}{2}}}\exp\left(\frac{k_{ik}m_{ik}N_{0}}{2\Omega_{ik}P_{\text{max}} s}\right)\\
&\times \mathcal{W}_{-\Delta-\frac{1}{2},\Theta}\left(\frac{k_{ik}m_{ik}N_{0}}{\Omega_{ik}P_{\text{max}} s}\right).
\label{dmk}
\end{align}
Thus, by appropriately substituting (\ref{ck}), (\ref{mk}) and (\ref{dmk}) into (\ref{ergcap}), ergodic capacity is obtained whenever $\frac{w}{|h_{jk}|^{2}}> P_{\text{max}}$ occurs in one or both hops.

Further, it is interesting to indicate the case when the latter condition occurs in terms of the distance between the primary and secondary nodes. By averaging out the presence of channel fading based on \cite[Eq. (5)]{ref3}, we have that
\begin{align*}
\frac{w}{|h_{jk}|^{2}}>P_{\text{max}}\Rightarrow d_{jk}^{-\alpha_{jk}}<\frac{w}{P_{\text{max}}}\Leftrightarrow d_{jk}>\left(\frac{P_{\text{max}}}{w}\right)^{\frac{1}{\alpha_{jk}}}.
\end{align*}
As an illustrative example, let $P_{\text{max}}=0$dB, $w=-3$dB and $\alpha_{jk}=4$ (i.e., a typical rural terrestrial). Then, only when the distance between the secondary source (or the relay) and the primary node is greater than $1.19$ km, the corresponding secondary node(s) may reach to the maximal achievable transmission power.

\section{Numerical Results and Discussion}
\begin{table}
\caption {Number of Terms Required For Convergence Up To The $4$th Decimal Point$^{*}$} 
\begin{center}
\begin{tabular}{l| l l l}\hline
$d_{\{j,1\}}=d_{\{j,2\}}$ & $w/N_{0}=0$dB &$w/N_{0}=10$dB & $w/N_{0}=15$dB\\\hline
 0.05 &49 &51 &51 \\
 0.1 &52 &56 &57 \\
 0.3 &57 &58 &58 \\
 0.8 &58 &59 &60 \\\hline
\end{tabular}
\end{center}
*$d_{\{i,1\}}=d_{\{i,2\}}=0.5$, $k_{i,l}=1$, $k_{j,l}=4$, $m_{i,l}=1$, $m_{j,l}=3$, $\alpha=4$, and $P_{\text{max}}/N_{0}=20$dB.
\label{table}
\end{table}

\begin{figure}[!t]
\centering
\includegraphics[trim=2.0cm 0.2cm 2.5cm 0cm,clip=true,totalheight=0.285\textheight]{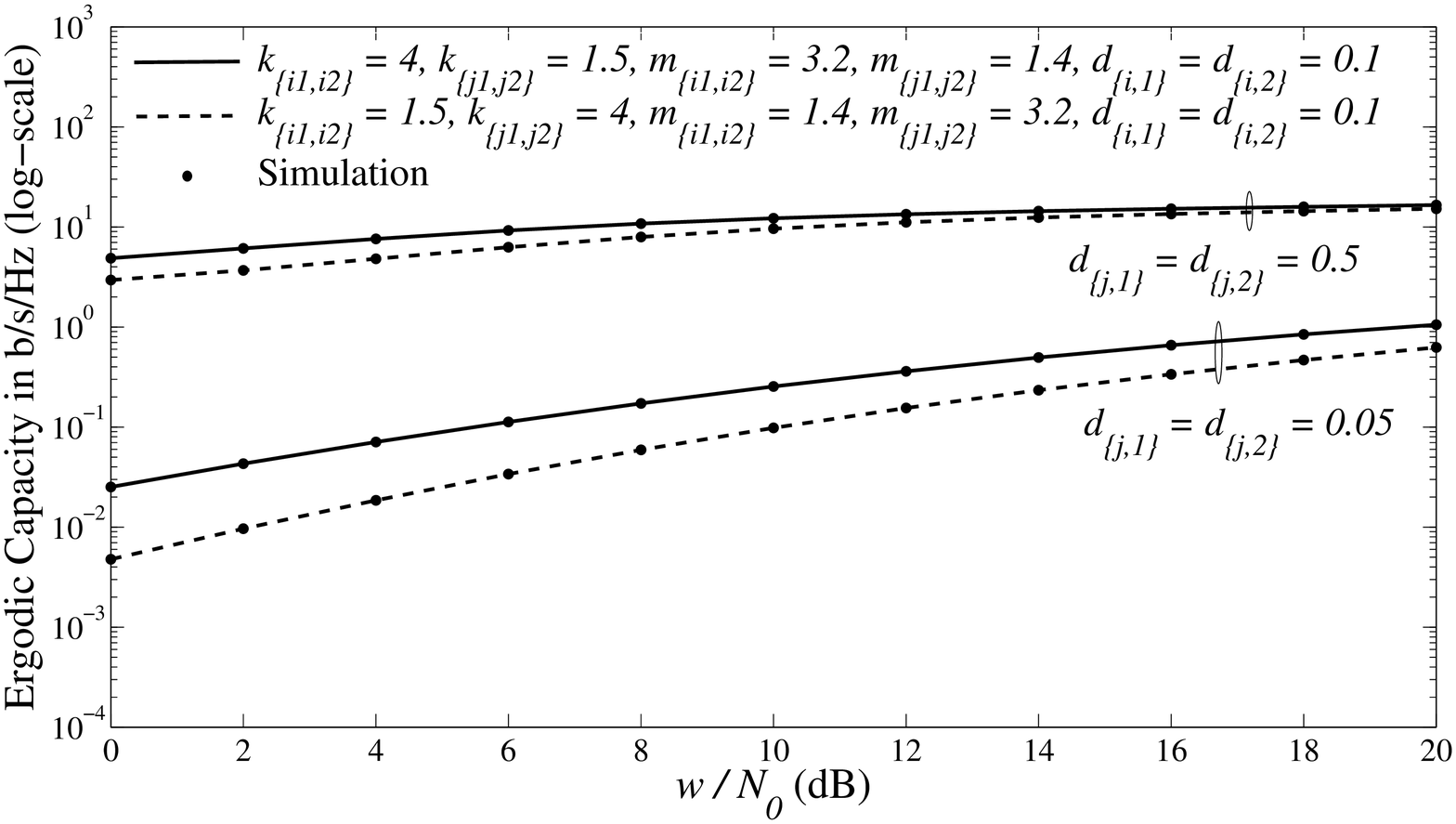}
\caption{Performance of the normalized $e2e$ ergodic capacity vs. various $w/N_{0}$ values.}
\label{fig2}
\end{figure}

\begin{figure}[!t]
\centering
\includegraphics[trim=2.5cm 0.2cm 2.5cm 0cm,clip=true,totalheight=0.285\textheight]{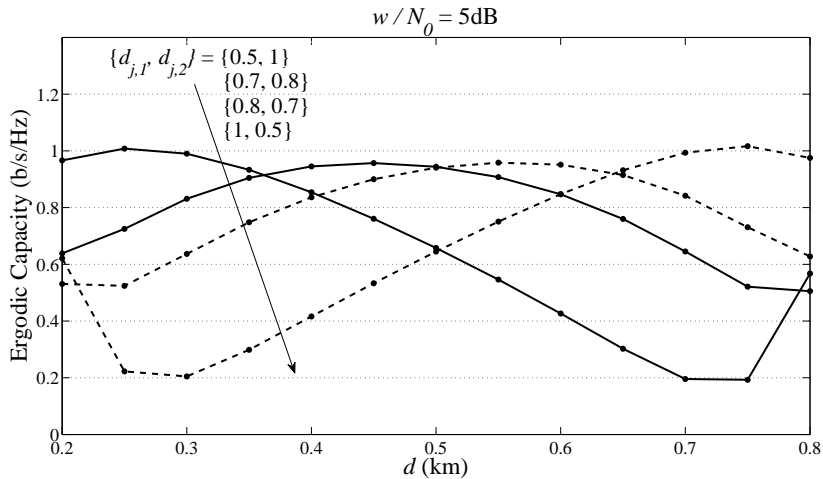}
\caption{Performance of the normalized $e2e$ ergodic capacity vs. various link distances between the secondary nodes, where $d_{i,1}\triangleq d$ and $d_{i,2}\triangleq 1-d$. Also, assume common fading parameters for all the involved links, namely, let $k=2$ and $m=1$, indicating the shadowing and multipath fading severity, respectively.}
\label{fig3}
\end{figure}

In this section, the theoretical results are presented (by setting $\mathcal{N}=60$) and compared with Monte-Carlo simulations. Table \ref{table} illustrates an indicative performance example in terms of the included series convergence. Similar behavior is observed for other system parameters. There is a good match between all the analytical and the respective simulation results and, hence, the accuracy of the proposed approach is verified. In what follows and without loss of generality, an identical path-loss factor is used for each link, namely, $\alpha=4$. 

In Fig. \ref{fig2}, symmetric link distances are assumed between the two hops. Obviously, fading severity affects the $e2e$ performance. However, it can be seen that the ergodic capacity is affected much more drastically by the link distance of P$_{R}$ (i.e., $d_{j,l}$) rather than the fading severity of the secondary signal.   

From a different standpoint, Fig. \ref{fig3} indicates how the ergodic capacity is influenced from various (non-symmetric) link distances. Curves in solid lines indicate closer primary-to-source distance, while the ones in dashed lines indicate closer primary-to-relay distance. It is clear that when the primary node is nearer (farther) to the secondary source (relay) or vice versa, then preserving symmetric distances between the secondary nodes (i.e., when $d=0.5$) is not a fruitful option.

\ifCLASSOPTIONcaptionsoff
  \newpage
\fi

\end{document}